\documentclass[12pt,reqno]{amsart}

\usepackage{amsthm}
\usepackage{amsmath}
\usepackage{latexsym}
\usepackage{amsfonts}
\usepackage{amssymb}
\usepackage{color}
\usepackage{bbm,dsfont}
\usepackage{graphicx}

\usepackage{subcaption}
\usepackage{mathrsfs}
\usepackage{mathbbol}
\usepackage{enumerate}
\usepackage{MnSymbol}
\usepackage{bbding}
\usepackage{multirow}
\usepackage[percent]{overpic}
\usepackage{etoolbox}

\usepackage{array}
\usepackage{makecell}
\usepackage[table]{xcolor}
\usepackage{appendix}
\usepackage{subcaption}
\usepackage[a4paper]{geometry}

\newtheorem{proposition}{Proposition}
\newtheorem{theorem}{Theorem}

\theoremstyle{definition}

\newcommand{\ket}[1]{| #1 \rangle}

\newcommand{\ketbra}[1]{| #1 \rangle \langle #1|}

\newcommand{\bea}{\begin{eqnarray}}
\newcommand{\eea}{\end{eqnarray}}

\newcommand{\no}[1]{\left\|#1\right\|} %norm
\newcommand{\tr}[1]{\mathrm{tr}\left[#1\right]} %trace
\newcommand{\id}{\mathbbm{1}} %identity operator
\newcommand{\parti}{\mathcal{P}} %partition
\newcommand{\M}{\mathsf{M}}
\newcommand{\N}{\mathsf{N}}
\newcommand{\C}{\mathsf{C}}
\newcommand{\mstd}{\mathsf{M}_{\textrm{std}}}

\newcommand{\cpost}{C_{\mathrm{post}}}
\newcommand{\cpre}{C_{\mathrm{pre}}}
\newcommand{\cmeta}{C_{\mathrm{meta}}}

\newcommand{\rpost}{R_{\mathrm{post}}}
\newcommand{\rpre}{R_{\mathrm{pre}}}
\newcommand{\rmeta}{R_{\mathrm{meta}}}

\newcommand{\ppre}{P_{\mathrm{pre}}}
\newcommand{\ppost}{P_{\mathrm{post}}}
\newcommand{\pmeta}{P_{\mathrm{meta}}}

\newcommand{\qxl}{q_{\ell}}
\newcommand{\rxl}{\varrho_{\ell}}

\makeatletter
\patchcmd{\@maketitle}
  {\ifx\@empty\@dedicatory}
  {\ifx\@empty\@dedicatory \par\vskip0.25em{\begin{center}\scriptsize\@setaddresses\end{center}}}
  {}{}
\makeatother

\title{Metainformation in quantum guessing games}
\author{Teiko Heinosaari \and Hanwool Lee}
\address{Faculty of Information Technology, University of Jyv\"askyl\"a, Finland}

\begin{document}

\maketitle
\makeatletter \let\@setaddresses\relax \makeatother

%%%%%%%%%%%%%%%%%%%%%%%%%%
\begin{abstract}

Quantum guessing games offer a structured approach to analyzing quantum information processing, where information is encoded in quantum states and extracted through measurement. An additional aspect of this framework is the influence of partial knowledge about the input on the optimal measurement strategies.
This kind of side information can significantly influence the guessing strategy and earlier work has shown that the timing of such side information, whether revealed before or after the measurement, can affect the success probabilities.
In this work, we go beyond this established distinction by introducing the concept of metainformation. 
Metainformation is information about information, and in our context it is knowledge that additional side information of certain type will become later available, even if it is not yet provided. We show that this seemingly subtle difference between having no expectation of further information versus knowing it will arrive can have operational consequences for the guessing task. Our results demonstrate that metainformation can, in certain scenarios, enhance the achievable success probability up to the point that post-measurement side information becomes as useful as prior-measurement side information, while in others it offers no benefit. 
By formally distinguishing metainformation from actual side information, we uncover a finer structure in the interplay between timing, information, and strategy, offering new insights into the capabilities of quantum systems in information processing tasks.
\end{abstract}

%%%%%%%%%%%%%%%%%%%%%%%%%%%%%%%%%%%%%%%%%%%
\section{Introduction}
%%%%%%%%%%%%%%%%%%%%%%%%%%%%%%%%%%%%%%%%%%

At its core, both classical and quantum information processing aim to produce a specific, desired output from a given input. This objective can be abstracted as a kind of game, where success is quantified by a scoring rule: correct outcomes receive high scores, while incorrect ones receive none. Framing information processing in this way—as a structured guessing challenge—provides a unifying perspective for diverse tasks, making it a powerful tool for analyzing the potential benefits that quantum systems offer in processing information.
Depending on the context, the guessing game may represent a communication task: for example, Alice may attempt to convey information to Bob, potentially while concealing it from others. Alternatively, it can describe a computational setting, where Alice selects an input string and performs a computation—possibly with Alice and Bob being the same agent.

Our focus is on quantum guessing games, a setting in which the information to be guessed is encoded into quantum states and later extracted through quantum measurements. An essential aspect in the current investigation is the presence of classical side information, which is an additional data that provides partial knowledge about the original input. 
This side information does not fully determine the input but can significantly influence the guessing strategy.
A central question in this context is when the side information is made available. Prior studies \cite{BaWeWi08,GoWe10,CaHeTo18,heinosaari2023anticipative} have demonstrated that the timing of this information—whether it is revealed before or after the quantum measurement— can make a crucial difference. This temporal distinction is not merely a technical detail; it can fundamentally alter the structure of optimal strategies and the achievable success probabilities \cite{carmeli2022quantum}.
One practical application of this sensitivity to timing is in the task of incompatibility witnessing \cite{CaHeTo19,CaHeMiTo19JMP,UoKrShYuGu19,SkSuCa19}.

In this work, we uncover a finer structure in the role of side information. Specifically, even when the side information is revealed after the measurement, there are two fundamentally different scenarios. In one case, the receiver is aware in advance that certain kind of additional information will be provided later. In the other, the receiver has no expectation of any further information. This distinction turns out to be operationally significant, as it can affect the strategies used and the achievable performance in the guessing task.
We demonstrate that the difference in these two scenarios can sometimes affect the optimal guessing probabilities, but not always.

Knowledge that additional information will later become available is not useful in isolation unless that side information is actually provided. We refer to this as metainformation, as it constitutes information about information. Distinguishing metainformation from actual side information allows us to clarify important distinctions in different operational scenarios.

It should be noted that the role of timing of side information in guessing games becomes crucial only when quantum systems are used as information carriers, as in the classical setting side information revealed before or after the measurement are equally good.
In fact, in classical setting the optimal measurement is always the measurement that perfectly distinguishes the pure dit states, hence there is no advantage that can be gained by changing the measurement.  

Our investigation is structured as follows. 
In Section \ref{sec:setup} we give a precise definition of metainformation and explain its difference to side information. A systematic way to calculate minimum error in different settings is presented in Section \ref{sec:min}. An alternative quantification of performance in guessing games is maximum confidence and in Section \ref{sec:max} we provide tools to tackle the questions with that measure of success.
Interestingly, post-measurement side information can always reach the same confidence as pre-measurement side information if metainformation is provided.
Finally, in Section \ref{sec:summary} we summarize our findings.

%%%%%%%%%%%%%%%%%%%%%%%%%%%%%%%%%%%%%%%%%%%
\section{Side information and metainformation in guessing games}\label{sec:setup}
%%%%%%%%%%%%%%%%%%%%%%%%%%%%%%%%%%%%%%%%%%

We are considering a class of guessing games where the sender, Alice, chooses a symbol $x$ from a finite label set $X$, prepares a quantum state $\varrho_x$ based on that choice and then transmits the quantum system to the receiver, Bob. 
We assume that Bob knows the encoding of symbols into states.
Bob is extracting information by performing a measurement on the received system.
His task is to make some decision that is based on $x$.
One obvious game is that Bob needs to guess the symbol $x$ that Alice chose. 
This task is referred as quantum state discrimination.
It can also be a different kind of game, where Bob e.g. needs to guess a symbol different from Alice choice.
That kind of task is referred as antidiscrimination, or exclusion.

In the current work, we investigate the case that Bob gets additional information on the symbol $x$.
We assume that this \emph{side information} is classical information and it conveys some partial knowledge about $x$.
A paradigmatic example of side information is that Bob is told that $x$ belongs to some subset $X' \subset X$.
Depending on which phase Bob gets the side information, it is either 1) \emph{pre-measurement information}, or 2) \emph{post-measurement information}.
Namely, the side information has to be given to Bob either before or after he has to perform the measurement, and that divides the side information into the mentioned two classes.

It is clear that pre-measurement information is at least as good aid for Bob as post-measurement information.
The reason is that Bob can choose to use the obtained side information later, which degrades pre-measurement side information effectively into post-measurement side information.
Further, it should not come as a surprise that in some cases pre-measurement side information is strictly more useful than post-measurement side information as the first one can be used in the choice of the measurement, while the latter one can be only used in the post-processing of the obtained measurement outcome.

The main focus of our current investigation is to recognize the role of an additional layer of information in the previously described setup.
Namely, in the setting of post-measurement information there are two distinct alternatives.
In the first case, Bob does not know that he will receive side information until he gets it.
Or to be more precise, he does not know what form of side information he will get.
After Bob has received side information, he uses it in the best possible way to carry out the task, which can be e.g. state discrimination. 
For instance, if Bob records a measurement outcome $y$ but post-measurement information rules out that specific symbol as a correct answer, then Bob is going to guess some other symbol.
In the second case, Bob knows that he will get post-measurement side information and he also knows the form of that information. 
In this latter case Bob has information about the side information, and therefore that is referred as \emph{metainformation}.

To give an example, assume that the message is chosen from a label set $X$, which is divided into a partition, $X= X_1 \cup X_2$. 
The side information is the correct part of the partition, i.e., the index $i$ for which $x\in X_i$.
As explained before, the side information can be transmitted to Bob before or after he is performing the measurement.
The metainformation in this case is the partition itself.
In practice, this means that when metainformation is transmitted to Bob, he is told that the subsets forming the partition are $X_1$ and $X_2$ but he is  not told to which subset the symbol $x$ belongs to.

There are several general properties that are characteristic of metainformation.
First, while side information typically increases the probability of guessing the correct index even if the measurement would fail to give any information, metainformation alone is useless.
Metainformation becomes useful only together with post-measurement side information.
A more technical way to say this is that Bob has in the beginning a prior distribution of the correct index $x$, which would be the uniform distribution on $X$ if there is no additional knowledge.
The side information changes this prior distribution, while
metainformation alone does not change it.
Second, in the case of pre-measurement side information there is no additional benefit of getting metainformation.
The reason is that since metainformation is information about side information, the latter contains the first one.
Metainformation has value only if the side information is obtained at a later stage, and the only relevant time step in our considered scheme is the measurement process.

\begin{figure}[t]
    \centering

    % First row
    \begin{subfigure}[b]{0.48\textwidth}
        \centering
        \fbox{\includegraphics[width=\linewidth]{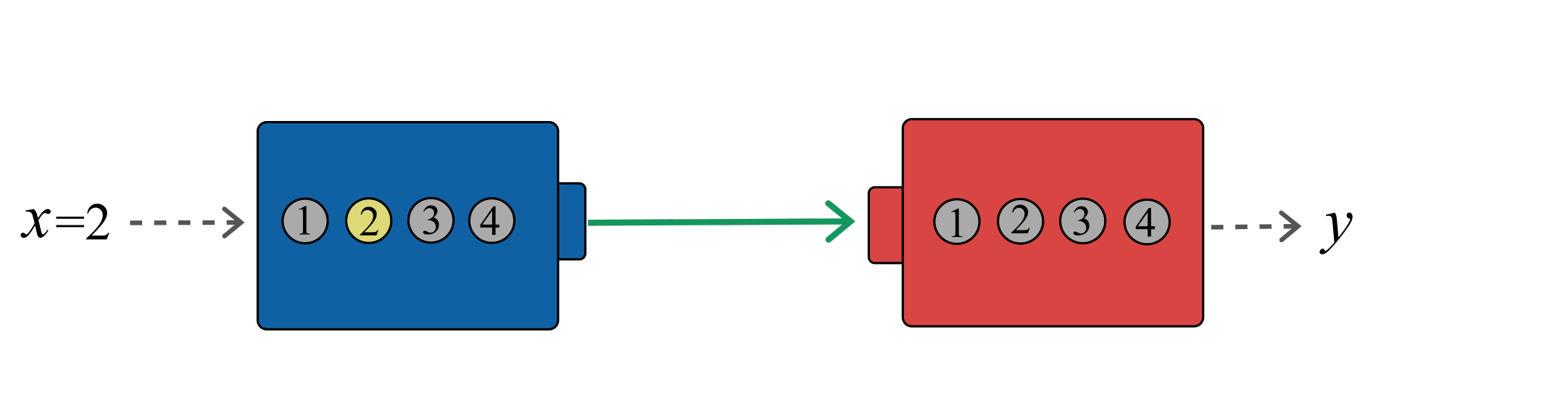}} % placeholder image
        \caption*{(a) No side information.}
    \end{subfigure}
    \hfill
    \begin{subfigure}[b]{0.48\textwidth}
        \centering
        \fbox{\includegraphics[width=\linewidth]{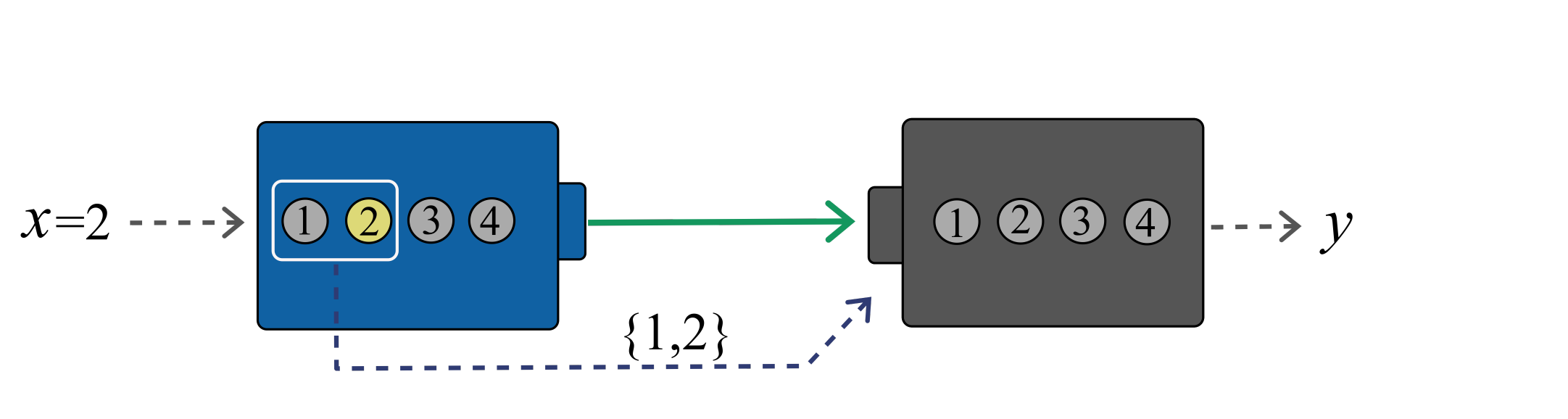}} % placeholder image
        \caption*{(b) Pre-measurement side information.}
    \end{subfigure}

    \vspace{1em}

    % Second row
    \begin{subfigure}[b]{0.48\textwidth}
        \centering
        \fbox{\includegraphics[width=\linewidth]{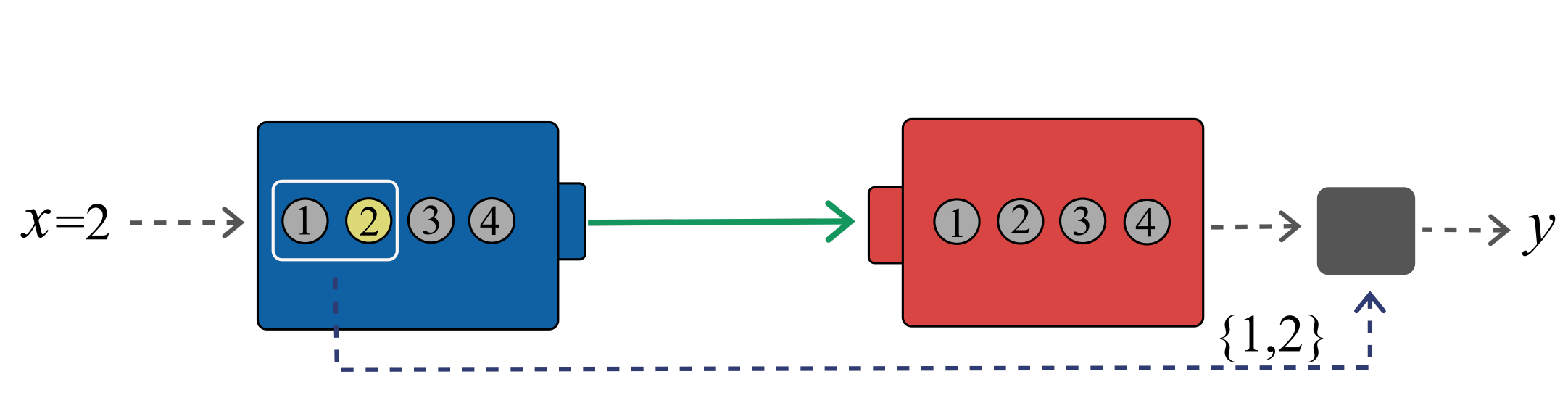}} % placeholder image
        \caption*{(c) Post-measurement side information.}
    \end{subfigure}
    \hfill
    \begin{subfigure}[b]{0.48\textwidth}
        \centering
        \fbox{\includegraphics[width=\linewidth]{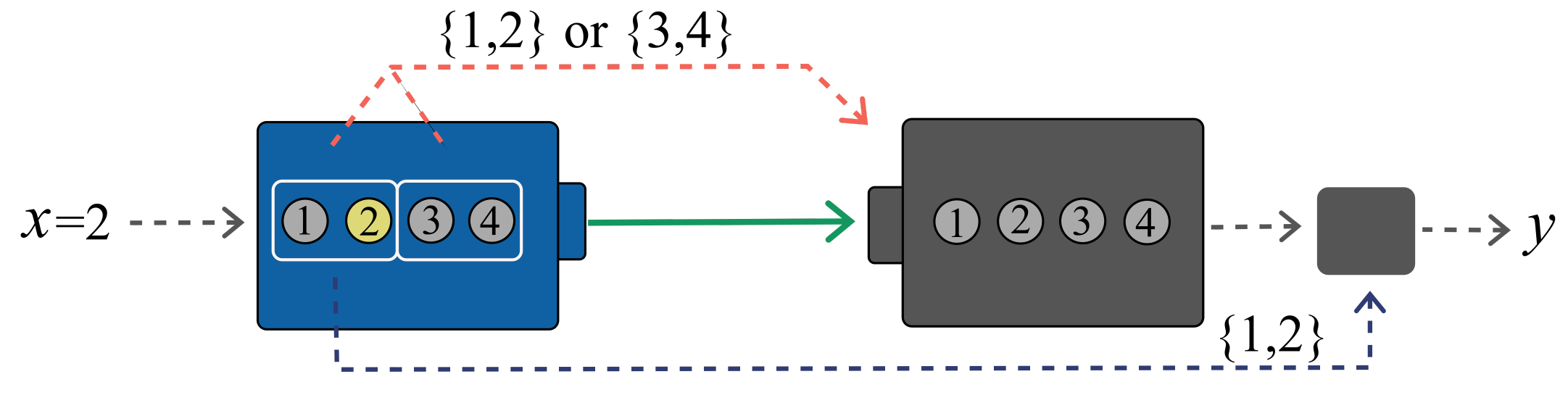}} % placeholder image
        \caption*{(d) Post with meta information.}
    \end{subfigure}

    \caption{Four different scenarios that differ in the ways how classical information complements the transmission of quantum information. The components that can be optimized according to the classical information are gray.}
    \label{fig:setup}
\end{figure}

From the previous discussion we conclude that there are three different scenarios when considering a guessing game with side information:
\begin{itemize}
\item pre-measurement side information
\item metainformation followed by post-measurement side information
\item post-measurement side information alone. 
\end{itemize}
The scenario that is also illustrative to have in our comparison is:
\begin{itemize}
\item no side information.
\end{itemize}
These four scenarios are depicted in Fig. \ref{fig:setup}.

Typically, we use the success probability $P$ of guessing the correct message as the figure of merit,  which is closely related to min-entropy \cite{konig2009operational}, and generally
\begin{equation}
\ppre \geq  \pmeta \geq \ppost \geq P \, ,
\end{equation}
where the subscripts indicate the form of side information. (Note that 'meta' means that both metainformation and post-measurement side information are given.)

Whether some of the above inequalities is equality or strict inequality depends on the form of side information. 
In the extreme case we can have $\ppre =P$, which means that side information is useless in its all types of delivery.
As we will next demonstrate, in typical cases side information is useful and then the question is how much the mentioned scenarios differ.
Depending on the guessing task and side information, it can happen, for instance, that $\ppre =  \pmeta$ or even $\ppre = \ppost$. 

Let us make the previous discussion concrete by considering a qubit system as the information carrier and certain types of side information.
The task of Bob is to guess the index that Alice has chosen.
The label set is $X=\{1+,1-,2+,2-\}$, whose symbols are chosen uniformly, and the information is encoded into four qubit states that form two orthonormal bases, separated by an angle $0<\theta < \frac{\pi}{2}$. 
Hence, the encoding is
\begin{equation}
    \rho_{b,p}=\frac{1}{2} \left( \id+p((-1)^{b-1} \sin(\tfrac{\theta}{2})\sigma_x+\cos(\tfrac{\theta}{2})\sigma_z) \right), 
\end{equation}
where $b\in \{ 1,2\}$ is called the \emph{value of the basis} and $p\in\{+,-\}$ is called the \emph{value of the parity}. 
The label set $X$ can be partitioned by basis, $X=X_1 \cup X_2$, or  by parity, $X=X_+ \cup X_-$, where $X_b=\{ b+,b- \}$ and $X_p=\{ 1p,2p \}$.
The classical side information is the value $i$ of the partition $X_i$, and the metainformation is the information about the chosen partitioning; see Fig. \ref{fig:2bases}.

\begin{figure}[t]
  \centering
  \begin{subfigure}[b]{0.45\textwidth}
    \centering
    \includegraphics[width=\linewidth]{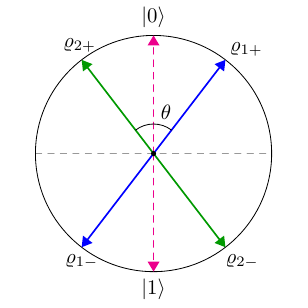}
    \caption*{(a) Basis side information.}
  \end{subfigure}
  \hspace{0.04\textwidth}
  \begin{subfigure}[b]{0.45\textwidth}
    \centering
    \includegraphics[width=\linewidth]{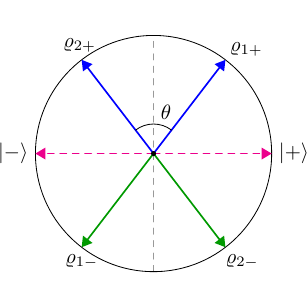}
    \caption*{(b) Parity side information.}
  \end{subfigure}
  \caption{Four symbols are encoded in qubit states that form two bases. In these two cases the discrimination tasks are the same, but they differ in the form of side information. The partitions that define the form of side information are marked by the color (green or blue). The measurements that are optimal for post-measurement side information with metainformation are represented by purple dashed arrows.}
  \label{fig:2bases}
\end{figure}

In order to find the optimal way to read the encoded quantum information, we denote $\M(x)$ the effect that concludes the symbol $x$ by a measurement $\M$. 
With no side information, the optimal measurement is any convex combination of projective measurements onto the bases, 
\begin{equation}\label{eq:measurement-without}
\M(b,p)=t_b \rho_{b,p} 
\end{equation}
where $t_1,t_2 \geq 0$ and $t_1+t_2=1$. Any of these measurements yield the success probabilty of $P=\frac{1}{2}$ \cite{schumacher2010quantum}. 

We will later present a systematic way to calculate the best use of side information with and without metainformation, but here we describe the optimal solutions in the previous two cases to demonstrate the role that metainformation has. 
The details of the calculations are explained in Sec. \ref{sec:min}. 
Let us first consider the case where $X$ is partitioned into $X=X_1 \cup X_2$ according to the value of the basis.  
If the value of the basis is known to Bob prior to measurement, then the success probability is $\ppre=1$ since Bob will discriminate orthogonal states. If Bob knows that he will get the information about the basis after the measurement, then the optimal measurement is the projective measurement related to $\sigma_z$, and the success probability is $\pmeta=\frac{1}{2}(1+\sqrt{\frac{1+|\cos\theta|}{2}} )$. 
Finally, if Bob does not know that he will receive post-measurement side information, then he performs the measurement \eqref{eq:measurement-without} and later he will post-process the measurement outcome with the classical information. 
In this case, the best success probability is $\ppost=\frac{1}{4}(3+|\cos\theta|)$. 
Therefore,  
\begin{equation}\label{eq:strict}
\ppre > \pmeta>\ppost>P
\end{equation}
and we conclude that both the timing of side information and the related metainformation are relevant.
Their effectiveness varies with $\theta$, see Fig. \ref{fig:plots} (a).

The strict ordering \eqref{eq:strict} may seem so natural that one wonders if it is always the case.
Let us consider the other mentioned partition, $X=X_+\cup X_-$, in which the side information is the value of the parity. 
When the value of parity is known before making a measurement, the optimal measurement is the projective measurement related to $\sigma_x$.

The success probability is $\ppre=\frac{1}{2}(1+\sin\frac{\theta}{2})$. 
It turns out that the same measurement is optimal also when metainformation and post-measurement information are given. 
Further, $\pmeta=\frac{1}{2}(1+\sin\frac{\theta}{2})$, which is the same as $P_{pre}$. Lastly, without metainformation, the success probability with post-measurement information is $\ppost=\frac{1}{4}(3-\cos\theta)$. 
Therefore, in this case we have 
\begin{equation}
\ppre = \pmeta>\ppost>P 
\end{equation}
and the success probabilities as functions of $\theta$ are depicted in Fig. \ref{fig:plots} (b).
We conclude that in this case pre-measurement information has no advantage over post-measurement information when the latter comes with metainformation.

\begin{figure}[t]
  \centering
  \begin{subfigure}[b]{0.47\textwidth}
    \centering
    \includegraphics[width=\linewidth]{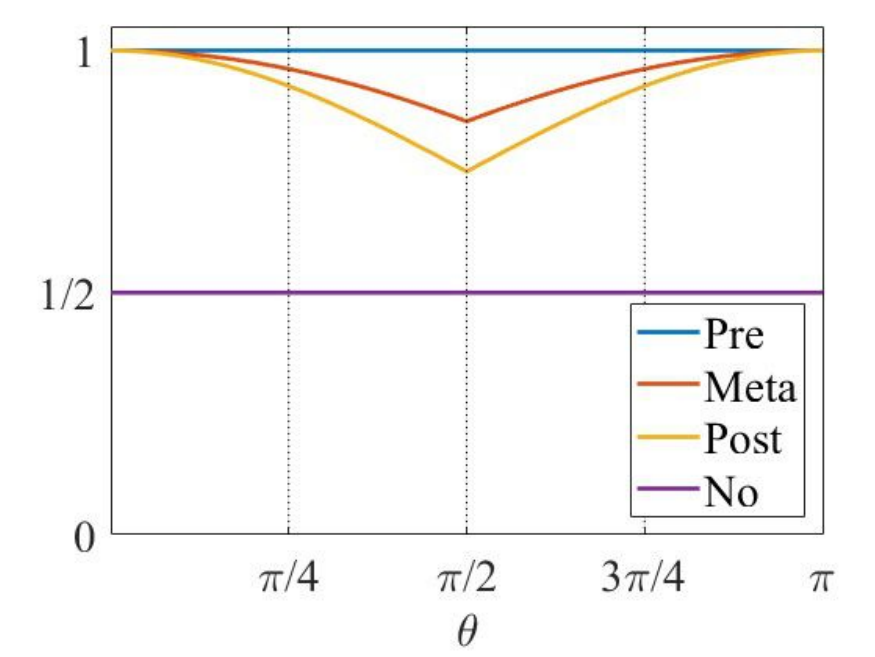}
    \caption*{(a) Basis side information.}
  \end{subfigure}
  \hspace{0.04\textwidth}
  \begin{subfigure}[b]{0.47\textwidth}
    \centering
    \includegraphics[width=\linewidth]{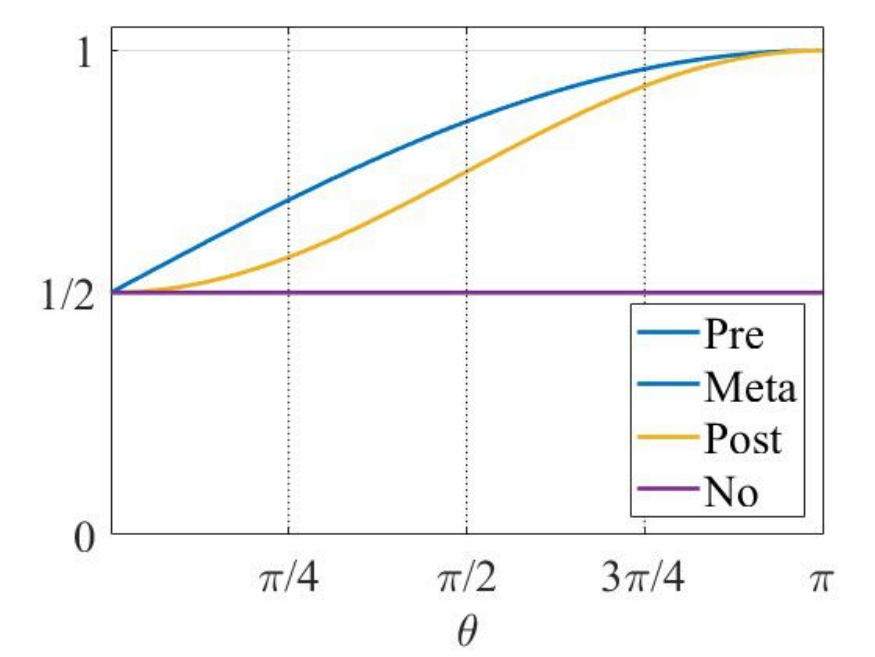}
    \caption*{(b) Parity side information.}
  \end{subfigure}
  \caption{The success probabilities of different side information scenarios in Fig. \ref{fig:2bases} are plotted as functions of $\theta$. In (a), the side information is the value of the basis. For this encoding, there is a hierarchy $\ppre>\pmeta>\ppost>P$ for all $0<\theta<\pi $. In (b), the side information is the value of parity. In this case, $\ppre=\pmeta$  for all
  $0<\theta<\pi$.}
  \label{fig:plots}
\end{figure}

%%%%%%%%%%%%%%%%%%%%%%%%%%%%%
\section{Minimum-error discrimination tasks}\label{sec:min}
%%%%%%%%%%%%%%%%%%%%%%%%%%%%

In this section we consider minimum error discrimination with side information.
As before, we denote by $X$ the set of symbols. 
A symbol $x$ is chosen from $X$ according to the prior probability $q_x$, then the corresponding quantum state $\varrho_x$ is prepared and sent to the receiver.
The side information is related to a partition $X=\cup_\ell X_\ell$ of $X$ into $m$ nonempty disjoint subsets.
We denote by $\parti$ the partition.
The side information, either delivered before or after the measurement, is the subset $X_\ell$ for which $x \in X_\ell$.
The corresponding metainformation is the partition $\parti$ itself.

%%%%%%%%%%%%%%%%%%%%%%%%%%%%%%%%%%%
\subsection{Pre-measurement information} 
%%%%%%%%%%%%%%%%%%%%%%%%%%%%%%%%%%%

In the case when side information is delivered before the measurement, the receiver knows that the sent symbol belongs to the subset $X_\ell$. 
Therefore, the task is to discriminate states from the subset $\{ \varrho_x : x \in X_\ell \}$ and the measurement can hence depend on the index $\ell$ that specifies the subset.
To find the optimal measurement for a given $\ell$ reduces to the usual state discrimination task.
One has to take into account that the a priori probabilities $q_x$ are now updated into $q'_x := q_x/(\sum_{y \in X_l} q_y)$.
The overall success probability $\ppre$ is the weighted average of the success probabilities of all $\ell$-dependent success probabilities.

To be specific, suppose a symbol $x \in X_\ell$ is chosen. 
The subset $X_\ell$ will be known to Bob before a measurement, and he will make the optimal measurement to discriminate all symbols in $X_\ell$, which can be found by the following optimization
\bea
P_\ell=\max_\M \sum_{x \in X_\ell} q_x\tr{\varrho_x \M(x)} \, .
\eea
Here $P_\ell$ is the success probability of discriminating the symbols in the subset $X_\ell$. 
The overall success probability is then
\bea
\ppre=\sum_\ell q_\ell P_\ell \, ,
\eea
where $q_\ell \equiv \sum_{x \in X_\ell} q_x$.

%%%%%%%%%%%%%%%%%%%%%%%%%%%%%%%%%%%
\subsection{Post-measurement information without metainformation}
%%%%%%%%%%%%%%%%%%%%%%%%%%%%%%%%%%%

In the case of post-measurement information without metainformation, Bob is not aware that he is going to receive side information and therefore he is performing the measurement that is optimal for the standard discrimination task. We call this measurement as the standard measurement and denote it by $\mstd$. The standard measurement is found by solving the optimization problem
\bea
\sum_{ x \in X} q_x\tr{\varrho_x \mstd(x)} = \max_{\M} \sum_{ x \in X} q_x\tr{\varrho_x \M(x)} \, .
\eea

When receiving side information after measurement, Bob can relabel the obtained measurement outcome and the relabeling can depend on the side information $X_\ell$.
We denote by $f_\ell$ a relabeling map that is applied in that case.
Hence, if the obtained  measurement outcome is $y$ and the side information is $X_\ell$, Bob will guess $f_\ell(y)$.
The success probability is then
\begin{align} 
\ppost &=\sum_\ell \sum_{x \in X_\ell} \sum_{y \in Y} q_x\tr{\varrho_x  \mstd(y)} \delta_{x,f_{\ell}(y)} \nonumber \\
& = \sum_\ell \sum_y q_{f_\ell (y)} \tr{ \varrho_{f_\ell (y)}\mstd(y) } \, .
\end{align}

From this expression we can find out what relabelling maps Bob should use.
For fixed $\ell$ and $y$, we choose $x=f_\ell(y)$ such that the expression
$q_x \tr{\rho_x \mstd(y)}$ takes the maximal value.
Therefore, the optimal relabeling map $f_\ell$ is defined as
\bea
f_\ell(y)=\arg \max_{x\in X_\ell} q_x\tr{\varrho_x \mstd(y)} \, .
\eea
In the case of multiple maximal values the relabeling map is not unique.

%%%%%%%%%%%%%%%%%
\subsection{Post-measurement information with metainformation}\label{sec:meta}
%%%%%%%%%%%%%%%%%

If metainformation is provided, the receiver knows the partition $X=\cup_\ell X_\ell$ before he performs the measurement.
A discrimination strategy then consists of a measurement $\M$ with an outcome set $Y$ and relabeling maps $f_\ell:Y \to X_\ell$ for each $\ell=1,\ldots,m$.
The expression for success probability is
\begin{align}
\pmeta 
& = \max_{\M} \max_{f_\ell} \sum_\ell \sum_{x\in X_\ell}\sum_y q_x \tr{\varrho_x\M(y) } \delta_{x,f_\ell(y)} \nonumber \\
& = \max_{\M} \max_{f_\ell}\sum_\ell \sum_y q_{f_\ell (y)} \tr{ \varrho_{f_\ell (y)}\M(y) } \, ,
\end{align}
where the maximization is taken over all possible measurements and relabelling maps. 
The difference to the previous case is that $\M$ can be something else than the standard measurement since the choice of $\M$ can depend on the metainformation.
The outcome set of $\M$ is, in principle, arbitrary. 
However, it has been shown in \cite{CaHeTo18} that we can restrict to the measurements that have the Cartesian product $X_1 \times \cdots \times X_m$ as their outcome space and to the relabeling functions that are the projections of $X_1 \times \cdots \times X_m$ onto $X_\ell$.
More precisely, for any choice of $\M$ and $\{f_\ell\}$, there is a measurement $\C$ on $X_1 \times \cdots \times X_m$ that, with the projection relabeling functions, gives the same success probability.

Intuitively, the receiver chooses a measurement that gives an outcome for any possible value of side information.
Then, when receiving the post-measurement side information, he picks the respective outcome.

By using $\C$ with the above properties, the success probability takes the form
\begin{align}
\pmeta = \max_{\C} \sum_{x_1,\ldots,x_m} \sum_{\ell} q_{x_\ell}\tr{  \varrho_{x_\ell} \C(x_1,\ldots,x_m)} \, ,
\end{align}
where $x_1 \in X_1$, $x_2 \in X_2$, and so forth.
This formula leads to the second useful observation \cite{CaHeTo18}. 
We define \emph{auxiliary states} as
\begin{align}\label{eq:aux}
\widetilde{\varrho}_{x_1,\ldots,x_m} = \frac{1}{\Delta_{x_1,\ldots,x_m}} \sum_{\ell} q_{x_\ell} \varrho_{x_\ell} \, ,
\end{align}
where the normalization factor is $\Delta_{x_1,\ldots,x_m}=\sum_{\ell} q_{x_\ell}$.
We have
\begin{equation}
\sum_{x_1,\ldots,x_m}  \Delta_{x_1,\ldots,x_m} =\sum_{x_1,\ldots,x_m}\sum_{\ell} q_{x_\ell} = |X_1| \cdots |X_m| \sum_\ell \frac{q_l}{|X_l|} \equiv  \Delta \, .
\end{equation}
Hence, $q_{x_1,\ldots,x_m} \equiv \Delta_{x_1,\ldots,x_m} / \Delta$ is a probability distribution.
Then, $\pmeta$ can be written as
\begin{align}
\pmeta = \Delta \cdot \max_\C \sum_{x_1,\ldots,x_m} q_{x_1,\ldots,x_m} \tr{  \widetilde{\varrho}_{x_1,\ldots,x_m} \C(x_1,\ldots,x_m)} \, .
\end{align}
The expression in the right hand side, in maximization, is the usual minimum error discrimination task of the auxiliary states.
In this way, the problem reduces to solving minimum error discrimination of auxiliary states, with the corresponding apriori distribution.
Following \cite{heinosaari2023anticipative}, we call the optimal measurement for the discrimination task with metainformation followed by the post-measurement side information as the \emph{anticipative measurement}.

%%%%%%%%%%%%%%%%%
\subsection{Basis encoding scheme}
%%%%%%%%%%%%%%%%%

We will consider a class of encoding schemes onto qubits, called basis encoding, and investigate when metainformation with post-measurement information gives advantage over just post-measurement information alone. 
In this scheme we encode an even number of symbols into qubit states that form several orthonormal bases.
The label set consists of pairs $1+$ and $1-$, $2+$ and $2-$, and so on.
Suppose there are $k$ such pairs and hence $2k$ symbols in total.

To encode the symbols into qubit states, we fix $k$ orthonormal bases. 
For each basis, which we denote as $\mathcal{B}_i$, we label the vectors as $+$ (positive direction) and  $-$ (negative direction). 

We denote $\theta_{ij}$ as the angle between the positive directions of $\mathcal{B}_i$ and  $\mathcal{B}_j$ in the Bloch ball; see Fig. \ref{fig:bases-encoding}. 

\begin{figure}
    \centering
    \includegraphics[width=0.5\linewidth]{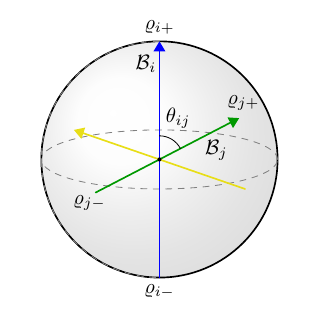}
    \caption{In the basis encoding scheme information is encoded into $k$ qubit bases. These can be depicted as central axes in the Bloch sphere. A label $i\pm$ is encoded to the quantum state $\varrho_{i\pm}$. The states $\varrho_{i+}$ and $\varrho_{i-}$ form an orthonormal basis.}
    \label{fig:bases-encoding}
\end{figure}

Lastly, we denote $\vec{n}_i$ the Bloch vector in the positive direction of $\mathcal{B}_i$. 
In summary, the label set is $X=\{i\pm\}_{i=1}^k$ and a label $i\pm$ is encoded to the quantum state
\begin{equation}
    \varrho_{i\pm}=\frac{1}{2}( \id \pm \vec{n}_i \cdot \vec{\sigma}) \, .
\end{equation}
The efficiency of this encoding depends on the angles $\theta_{ij}$ and we will soon find out the details.

The optimal measurement without classical side information is given as 
\begin{equation}
    \mstd(i\pm)=t_i \varrho_{i\pm} 
\end{equation}
where the coefficients $t_i \geq 0$ satisfy $\sum_{i=1}^k t_i =1$, but can be otherwise chosen freely \cite{schumacher2010quantum}.
Any of these measurements yield the success probability $P=\frac{1}{k}$. This means that any probabilistic choice of projective measurements onto the states $\{\varrho_{i+},\varrho_{i-} \}$ is optimal strategy. 
In the following, however, we will show that not all choices of projective measurements are the optimal when post-processing with classical side information is considered. 
Namely, only certain projective measurements will give the optimal success probability with post-processing.

We consider two different kinds of side information, i.e., ways of partitioning the label set. 
The first case is when the side information is the value $i$ of the basis $\mathcal{B}_i$, and the partition is $X=X_1\cup X_2\cup \ldots \cup X_k$ where $X_i=\{i+, i-\}$. 
In order to calculate the success probability when post-measurement side information comes with metainformation, we form the auxiliary states as explained in Subsec. \ref{sec:meta}.
Let us denote $\widetilde{\rho}_{\vec{p}}$ the auxiliary state, where $\vec{p}=(p_1,\ldots , p_k)$ and $p_i = \pm$. 
There are in total $2^k$ auxiliary states and using \eqref{eq:aux} they turn out to be
\bea\label{eq:aux-p}
\widetilde{\varrho}_{\vec{p}}=\frac{1}{k}\sum_{i=1}^k \varrho_{i,p_i}=\frac{1}{2}( \id+\frac{1}{k}\sum_{i=1}^k p_i \vec{n}_i \cdot \vec{\sigma}) \, .
\eea
The corresponding apriori probability distribution $q_{\vec{p}}$ is uniform and $\Delta=2^{k-1}$.

The second case is when the side information is the value of parity, and the partition is $X=X_{+}\cup X_-$ where $X_+=\{i+\}_{i=1}^k$ and $X_-=\{i-\}_{i=1}^k$.
There are now $k^2$ auxiliary states and they are labeled with two indices, $ij$,  where $i,j\in \{1,\ldots,k\}$.
The auxiliary states have the form
\bea\label{eq:aux-ij}
\widetilde{\varrho}_{ij}=\frac{1}{2}(\varrho_{i+}+\varrho_{j-})=\frac{1}{2}(\id+\frac{1}{2}(\vec{n}_i - \vec{n}_j)\cdot \vec{\sigma}) \, .
\eea
The corresponding apriori probability distribution $q_{ij}$ is uniform and $\Delta=k$.

In both cases above, the set of auxiliary states has an important property: for each state, there exists a state opposite to each other and having the same spectrum.
Namely, $\widetilde{\varrho}_{\vec{p}}$ and $\widetilde{\varrho}_{-\vec{p}}$ have the Bloch vectors which are opposite to each other and have the same length. 
Likewise, $\widetilde{\varrho}_{ij}$ and $\widetilde{\varrho}_{ji}$ have Bloch vectors which are opposite and of the same length. 
For such sets of states, one can find the optimal measurement analytically.

We first recall the following result.

\begin{proposition}[Prop.1 in \cite{carmeli2022quantum}] 
Consider a set of $n$ states $\{ \varrho_x\}_{x \in X}$ that are prepared with uniform apriori distribution.
We denote by $\lambda$ the largest eigenvalue of all those states, i.e., $\lambda:=\max_x \no{\varrho_x}$. 
Then, the success probability of discriminating these states is upper bounded as $P \leq \frac{d}{n} \lambda$. 
The upper bound is attained if and only if there exists a measurement $\M$ such that
\bea \label{eq:optcondition}
 \M (x) \varrho_x= \lambda \M(x) 
\eea
for all $x\in X$. In such a case, $\M$ is an optimal measurement. 
\end{proposition}

Applying the result to sets of qubit states with the mentioned specific symmetry leads us to the following conclusion.

\begin{proposition}\label{prop:qubitsym}
Let $\{ \varrho_x\}_{x \in X}$ be a set of $n$ qubit states such that for each $\varrho_x$ there is exactly one state $\varrho_{x'}$ having the opposite Bloch vector with equal length. Let us denote $\lambda$ the largest eigenvalue of all the states. Then the success probability of discriminating these states is $P=\frac{2}{n}\lambda$. 
The optimal measurement is any convex combination of projective measurements on the states $\varrho_x$ that have the largest eigenvalue.
\end{proposition}

\begin{proof}
We denote by $X^*\subset X$ the subset of those indices for which the corresponding state have the maximal eigenvalue, i.e., $X^*=\{ x \in X| \no{\varrho_x} = \lambda\}$.
By the assumption, the set $X^*$ has even number of indices and they come in pairs, $y$ and $y'$, such that $\varrho_y$ and $\varrho_{y'}$ have opposite Bloch vectors with equal length.

We denote $\Lambda_y$ the one-dimensional eigenprojection of $\varrho_y$ associated with $\lambda$. 
We can choose freely convex coefficients $t_y$ and define a measurement as 
$$
\M(y)=t_y \Lambda_{y}, \quad \M(y')=t_y \Lambda_{y'}
$$
and
$$
\M(x)=0\quad  \forall x \not \in X^* \, .
$$
By the assumption, $\Lambda_{y}+\Lambda_{y'}=\id$ and therefore $\sum_x \M(x)=\id$.
It is straightforward to verify that $\M$ satisfies \eqref{eq:optcondition}. 

In the other direction, from \eqref{eq:optcondition} follows that an optimal measurement $\M$ must have $\M(x)=0$ for all $x\notin X^*$ and $\M(y)$ is a scalar multiple of $\Lambda_y$ for all $y \in X^*$.
It follows from the normalization and that $\M$ is a convex combination of projective measurements on the states $\varrho_y$ with $y\in X^*$.
\end{proof}

Since the optimal measurements for state discrimination with meta information are the optimal measurements for the auxiliary states, we establish the following results.

\begin{theorem}\label{thm:basis} 
When the value of basis is side information, the success probability with  
post-measurement information is 
\begin{equation}\label{basis:post}
P_{\mathrm{post}}=\frac{1}{2}+\frac{1}{2k}\max_{i} \sum_{j=1}^k |\cos\theta_{ij}|
\end{equation}
and with metainformation
\begin{equation}\label{basis:meta}
P_{\mathrm{meta}}=\frac{1}{2}+\frac{1}{2k} \max_{p_1 \ldots p_k}\sqrt{\sum_{i,j=1}^{k} p_i p_j \cos\theta_{ij}}
\end{equation}
where the maximization is over all possible combinations of parity $p_i\in \{+,-\}$. 

\end{theorem}

\begin{proof}

We will first obtain $\ppost$. 
 Recall that the success probability with post-measurement information is computed as
\begin{align*}
\ppost=\sum_\ell \sum_y \tr{M_y q_{f_\ell (y)} \rho_{f_\ell (y)}} \, .
\end{align*}
with the relabeling map given as
\begin{align*}
    f_\ell(y)=\arg \max_{x\in X_\ell} q_x\tr{\rho_x \M(y)} \, .
\end{align*}
Let us rewrite the measurement outcome $y=(j,p_j)$ where $j$ is the basis value and $p_j$ is the parity value. The standard measurement is any convex combination of projective measurements onto the basis, so $\M(j,p_j)=t_j \rho_{j,p_j}$ where $t_j$ are convex coefficients. Note that $f_\ell((j, p_j))=(l,p^*)$, where $p^*$ is a parity on the basis $\mathcal{B}_\ell$ closer to $p_j$ on basis $\mathcal{B}_j$. Then, the success probability is
\begin{eqnarray*}
\ppost &=& \sum_\ell \sum_y \tr{M_y q_{f_\ell (y)} \rho_{f_\ell (y)}}= \frac{1}{2k} \sum_j t_j \sum_\ell \sum_{p_j=\pm 1} \tr{\rho_{j,p_j} \rho_{f_\ell ((j,p_j))}}\\
&=&\frac{1}{2k} \sum_j t_j \sum_\ell  (1+|\cos \theta_{j\ell} |) =\frac{1}{2}+\frac{1}{2k} \sum_j t_j \Theta_j
\end{eqnarray*}
where $\Theta_j=\sum_\ell | \cos\theta_{ \ell j} |$. The maximum value of $\ppost$ is attained when the measurement is projective measurements onto the basis $j^*=\arg \max_j \Theta_j$.

Now let us find $\pmeta$. 
The auxiliary states are given in \eqref{eq:aux-p}. According to Prop. \ref{prop:qubitsym}, we need to find the auxiliary states that have the largest eigenvalue. Denote $\vec{p}^*=(p_1^*,\ldots,p_k^*)=\arg \max_{\vec{p}} \no{\sum_i p_i \vec{n}_i}$. Since $\no{\sum_i p^*_i \vec{n}_i}=\no{-\sum_i p^*_i \vec{n}_i}$, one can find a projective measurement in the direction of $\sum_i p^*_i \vec{n}_i$ and $-\sum_i p^*_i \vec{n}_i$ that satisfies \eqref{eq:optcondition}. Namely, any convex combination of projective measurement defined as
$$
\M(\pm\vec{p}^*)=\frac{1}{2}(\id \pm\frac{\sum_i p^*_i \vec{n}_i }{\no{\sum_i p^*_i \vec{n}_i }} \cdot \vec{\sigma})
$$
is the optimal measurement for the auxiliary states. 
The success probability for the auxiliary states is then found by using Prop.\ref{prop:qubitsym},
$$
\widetilde{P}=\frac{2}{2^k} \no{\widetilde{\varrho}_{\vec{p}^*}} =\frac{1}{2^k} (1+\frac{1}{k} \max_{\vec{p}} \no{\sum_i p_i \vec{n}_i })=\frac{1}{2^k}(1+\frac{1}{k}\max_{\vec{p}}\sqrt{\sum_{i,j} p_i p_j \cos\theta_{ij}}) \, .
$$
Since $\Delta=2^{k-1}$, we finally obtain the success probability with metainformation
\begin{align*}
\pmeta=\Delta \cdot \widetilde{P}=\frac{1}{2}+\frac{1}{2k} \max_{p_1 \ldots p_k}\sqrt{\sum_{i,j=1}^{k} p_i p_j \cos\theta_{ij}} \, .
\end{align*}

\end{proof}

\begin{theorem}\label{thm:parity}
When the value of parity is the side information, the success probability with post-measurement information is 
\begin{equation}\label{parity:post}
P_{\mathrm{post}}=\frac{1}{k}(\frac{3}{2}-2\cos\theta^*).
\end{equation}
and with metainformation 
\begin{equation}\label{parity:meta}
P_{\mathrm{meta}}=\frac{1}{k}(1+\sin\frac{\theta^*}{2})
\end{equation}
where $\theta^*=\max_{ij} \theta_{ij}$.
 
We have $P_{\mathrm{meta}}>P_{\mathrm{post}}$ whenever $0<\theta^*<\pi $.
\end{theorem}

\begin{proof}
Let us first find $\ppost$. Denote $\ell\in \{1,-1\}$ the value of parity sent as a side information and a measurement outcome $y=(j,p_j)$. The standard measurement is $\mstd(j,p_j)=t_j \varrho_{j,p_j}$ where $t_j$ are convex coefficients. The optimal post-processing map is $f_\ell((j,p_j))=(j^*,\ell)$, where 
\begin{eqnarray*}
j^*&=&\arg\max_i \tr{\varrho_{i,\ell} \mstd(j,p_j)}=\arg \max_i t_j \tr{\varrho_{i,\ell} \varrho_{j,p_j}}\\
&=&\arg \max_i (-1)^{1-\delta_{\ell, p_j} }\cos\theta_{ij}
\end{eqnarray*}
This means that if $\ell=p_j$, the relabeling map will choose $j^*=j$. If $\ell \neq p_j$, then $j^*$ will be chosen such that it minimizes $\cos\theta_{ij}$, i.e., it maximizes $\theta_{ij}$. Thus, 
\begin{eqnarray*}
\ppost&=&\frac{1}{2k}\sum_{\ell=\pm} \sum_y \tr{\varrho_{f_\ell (y)} \M(y)}= \frac{1}{2k}
\sum_{\ell = \pm} \sum_j \sum_{p_j = \pm} t_j \tr{\varrho_{f_\ell ((j,p_j))} \varrho_{j,p_j}}\\
&=& \frac{1}{2k} \sum_j t_j (\tr{\varrho_{f_+(j+)} \varrho_{j+}}+\tr{\varrho_{f_+(j-)} \varrho_{j-}}+\tr{\varrho_{f_-(j+)} \varrho_{j+}}+\tr{\varrho_{f_-(j-)} \varrho_{j-}} )\\
&=&\frac{1}{2k}(2+2\sum_jt _j \tr{\varrho_{f_+(j-)} \varrho_{j-}})=\frac{1}{k}(\frac{3}{2}-\frac{1}{2}\sum_j t_j \cos\theta_{j j^*})
\end{eqnarray*}
The optimal measurement is the projective measurements on the basis $j$ that minimizes $\cos \theta_{j j^*}$. 
In other words, the whole problem is to find the largest angle $\theta_{ij}$ for all $i,j$. The optimal measurement is a convex combination of the projective measurements on all such $i$ and $j$. 

Now let us compute $\pmeta$. 
The auxiliary states are $\widetilde{\rho}_{ij}$ are given in \eqref{eq:aux-ij}. 
The anticipative measurements are the optimal measurements for the auxiliary states, which is any convex combination of projective measurements on the auxiliary states $\widetilde{\rho}_{ij}$ that have the largest eigenvalue.  The success probability of the auxiliary ensemble is
\begin{align*}
\widetilde{P}&=\frac{2}{k^2} \max_{ij} ||\widetilde{\varrho}_{ij}||
=\frac{1}{k^2}(1+\max_{ij}\frac{1}{2} ||\vec{n}_i - \vec{n}_k||)\\
&=\frac{1}{k^2}(1+\sin\frac{\theta^*}{2})
\end{align*}
where $\theta^*=\max_{ij}\theta_{ij}$. Since $\Delta=k$, we get
\begin{align*}
    \pmeta=\Delta \cdot \widetilde{P}=\frac{1}{k}(1+\sin\frac{\theta^*}{2}) \, .
\end{align*}

\end{proof}

\subsection{Examples}

\begin{figure}
  \centering
  \begin{subfigure}[b]{0.45\textwidth}
    \centering
    \includegraphics[width=\linewidth]{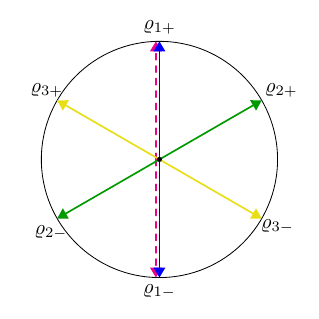}
    \caption*{(a)}
  \end{subfigure}
  \hspace{0.04\textwidth}
  \begin{subfigure}[b]{0.45\textwidth}
    \centering
    \includegraphics[width=\linewidth]{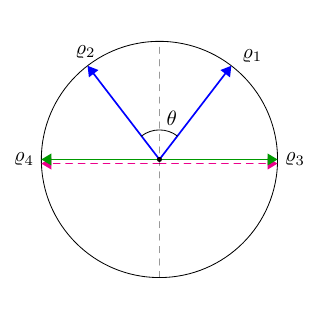}
    \caption*{(b)}
  \end{subfigure}
  \caption{Two examples that illustrate different inequalities in success probabilities with side information which is marked by different colors. Anticipative measurements are depicted as purple dashed arrows. In (a), encodings are equally separated bases in a plane. Since the anticipative measurement coincides with a standard measurement, $\ppre>\pmeta=\ppost$. In (b), the standard measurement is the same as the optimal measurement with pre-measurement information, and therefore $\ppre=\pmeta=\ppost$.  }
  \label{fig:min-examples}
\end{figure}

In Sec. \ref{sec:setup} we have seen that there are situations where $\ppre > \pmeta>\ppost>P$ and $\ppre = \pmeta>\ppost>P$. In the following we demonstrate that also other inequalities can turn into equalities, with suitable side information.

\subsubsection{The case $\ppre>\pmeta=\ppost>P$}

An interesting occurence is when side information is useful but metainformation is not.
Let us consider a symbol set $X=\{1\pm, 2\pm, 3\pm\}$ with the encoding 
$$
\varrho_{j\pm}=\frac{1}{2} \left(\id \pm  (\sin\left(\tfrac{j\pi}{3}\right) \sigma_x + \cos\left(\tfrac{j\pi}{3}\right)\sigma_z ) \right), j=1,2,3.
$$
These 6 states are uniformly separated in a plane; see Fig. \ref{fig:min-examples} (a). 
The optimal measurement is any convex combination of projective measurements onto the basis, which yields $P=\frac{1}{3}$. 
Consider the partition $X=X_1 \cup X_2 \cup X_3$ where $X_j=\{j\pm\}$ so the side information is the value of the basis. If the basis is known before the measurement, Bob will discriminate orthogonal states, so $\ppre=1$. 
By Theorem \ref{thm:basis}, we have
$\pmeta=\ppost=\frac{5}{6}$. 
Hence,
$$
\ppre>\pmeta=\ppost>P \, .
$$

\subsubsection{A case with $\ppre=\pmeta=\ppost>P$}

It can happen that post-measurement side information without metainformation is as useful as pre-measurement side information, i.e., $\ppre=\ppost$. 
To illustrate that occurence, consider a following label set with a partition $X=X_1 \cup X_2$, where the subsets consist of two states $X_1=\{1, 2\}, X_2=\{3, 4\}$ defined as 
$$
\rho_{1,2}=\frac{1}{2}(\id \pm\sin(\frac{\theta}{2}) \sigma_x + \cos(\frac{\theta}{2})\sigma_z) \, , \quad \rho_{3,4}=\frac{1}{2}(I\pm \sigma_x) ,
$$
see Fig. \ref{fig:min-examples} (b).
Without classical side information, the optimal measurement is $\sigma_x$ measurement and the success probability is $P=\frac{1}{2}$. 
In the pre-measurement information scenario, the optimal measurements for the two subsets $X_1$ and $X_2$ are the $\sigma_x$ measurement which yields $\ppre=\frac{3}{4}(1+\sin(\frac{\theta}{2}))$. In post-measurement information case, Bob makes the $\sigma_x $ measurement and guesses a state in the known subset $X_\ell$ that is closest to the effect corresponding to the measurement outcome. After this post-processing, the success probability is $\ppost=\frac{3}{4}(1+\sin(\frac{\theta}{2}))$. Therefore, 
$$
\ppre=\pmeta=\ppost>P \, .
$$

%%%%%%%%%%%%%%%%
\section{Maximum-confidence discrimination tasks}\label{sec:max}
%%%%%%%%%%%%%%%%%%

So far, we have used the success probability as our primary figure of merit. However, in the context of state discrimination, a more refined figure of merit is the concept of \emph{confidence} and the probability of achieving it. 
While success probability reflects the overall distinguishability of the ensemble, the confidence for a symbol $x$, denoted as $c(x)$, quantifies the amount of information about $x$ that a measurement extracts. Specifically, it is defined as the conditional probability that the outcome of a measurement correctly identifies the symbol $x$:
\begin{equation}
c(x)=\Pr(X=x|Y=x)=\frac{q_x\tr{\varrho_x \M(x)}}{\tr{\varrho \M(x)}} \label{eq:confidence}
\end{equation}
where $Y$ denotes measurement outcomes, $\varrho=\sum_x q_x \varrho_x$, and $\M(x)$ is an effect denoting an outcome $x$ of a measurement $\M$. A natural accompanying figure of merit is the probability with which this confidence is attained. We refer to the sum of these probabilities over all symbols as the \emph{conclusive rate}.

Interestingly, the figures of merits in different discrimination strategies such as minimum-error discrimination and unambiguous discrimination can be unified by confidence and the conclusive rate. For instance, the success probability can be written as the average confidence over the ensemble,
\begin{align*}
    P=\sum_{x\in X} q_x\tr{\varrho_x \M(x)}=\sum_{x \in X} c(x) \tr{\varrho \M(x)},
\end{align*}
where $\tr{\varrho \M(x)}$ is interpreted as the probability of obtaining confidence $c(x)$. On the other hand, a measurement realizes unambiguous discrimination if $c(x)=1$ for all $x\in X$. 
If the conclusive rate is maximized, this measurement achieves the optimized unambiguous discrimination \cite{chefles1998optimum}.

In this section, we consider maximum-confidence discrimination with side information. 
Maximum confidence measurement (MCM) is a measurement that maximizes confidence for all $x$. We denote the maximum confidence of $x$ as $C(x)$ and this can be obtained from the following optimization,
\bea
C(x)=\max_{\M(x)} \frac{ q_x\tr{ \varrho_x\M(x)}}{\tr{\varrho \M(x)}}. \label{eq:maxconfi}
\eea 
Note that while unambiguous discrimination can be realized only for restricted class of ensembles, such as linearly independent pure states, maximum confidence measurement can be found for any collection of states.
The analytic solution of \eqref{eq:maxconfi} is known in \cite{croke2006maximum,lee2022maximum},
\bea
C(x)=\no{\sqrt{\varrho}^{-1} q_x \varrho_x \sqrt{\varrho}^{-1}}. \nonumber
\eea

In the following, we denote $\N$ an MCM, i.e.,$\N(x)=\arg\max c(x)$. Note that $\N(x)$ is not unique, since for any positive scalar can be multiplied to $\N(x)$ and still gives the same confidence. The analytic solution of a MCM is,
\begin{equation*}
\N(x)\propto \sqrt{\varrho}^{-1}\Lambda(\sqrt{\varrho}^{-1} \varrho_x\sqrt{\varrho}^{-1})\sqrt{\varrho}^{-1} \, ,
\end{equation*}
where $\Lambda(\cdot)$ denotes the eigenprojections associated with the largest eigenvalue. 
Although it is not difficult to consider a mixed state, we mainly focus on the case when $\varrho_x$ is a pure state. 
In such a case, MCM can be written as
\bea
\N(x)=a_x \Pi_x \label{eq:mcm}
\eea
where $0\leq a_x\leq 1$ and 
\bea
\Pi_x=\frac{\varrho^{-1} \varrho_x \varrho^{-1}}{\tr{\varrho^{-1} \varrho_x \varrho^{-1}}} \label{eq:mcmprojector}
\eea
is a rank-1 projection. 
We call $\Pi_x$ a \textit{MCM projection} and it is uniquely determined by the state ensemble. 
Since a scalar $a_x$ does not affect confidence, $\Pi_x$ is the sole factor for the value of $C(x)$.  

In general, $\sum_{x \in X}\N(x)$ may not be the identity operator. 
When needed, we include an additional effect, which we denote as $\N(0)$, such that the collection of effects form a POVM, i.e., $\sum_{x \in X} \N(x)+\N(0)=\id$. 
One can further optimize a MCM $\N$ such that it maximizes the probability of obtaining the maximum confidence $C(x)$. We call it the optimal \textit{conclusive rate} $R$, which can be obtained via
\bea
R=\max_\N \tr{\varrho \sum_{x \in X} \N(x)} \, . \label{eq:conclusive}
\eea
When an ensemble consists of only pure states, the maximization is written as 
\bea
R=\max_{\{a_x\}}\sum_{x=1}^n a_x \tr{\varrho \Pi_x} \, .
\eea
We call a measurement the \textit{optimal} MCM if it maximizes both the confidence and the conclusive rate. 

When side information is provided, we need different formulas for confidence and the conclusive rate other than \eqref{eq:confidence} and \eqref{eq:conclusive}. Suppose a set of symbols is partitioned as $X=\cup_\ell X_\ell$ and the side information is the value of the subset $\ell$.  Given $\ell$ and the measurement outcome $y$ of a measurement $\M$, the conditional probability that a symbol $x$ is sent can be found by using the chain rule,
\begin{equation}
\Pr(x|y,\ell)=\frac{\Pr(x|\ell) \Pr(y|x,\ell)}{\Pr(y|\ell)}=\frac{q_x\tr{\varrho_x \M(y)}}{\qxl \tr{ \rxl \M(y) }} \label{eq:confiside}
\end{equation}
where $\qxl=\sum_{x\in X_\ell} q_x$ and $\rxl=\frac{\sum_{x \in X_\ell} q_x \varrho_x}{\qxl}$. This quantity will be used for calculating the maximum confidence with side information. 
In addition, different measurement outcomes can be considered as conclusive when side information is taken into account.

%%%%%%%%%%%%%%%%
\subsection{Pre-measurement information}
%%%%%%%%%%%%%%%%%%

Let us first find MCM with pre-measurement information. If the subset $X_\ell$ is known before measurement, Bob can simply make a measurement that maximizes confidence for all $x \in X_\ell$. 
The maximum confidence with pre-measurement information is obtained from the optimization by using \eqref{eq:confiside},
\bea
\cpre(x)&=&\max_{\M(x) \geq 0} \frac{q_x\tr{\varrho_x \M(x)}}{\qxl \tr{ \rxl \M(x)}}=q_x \qxl^{-1} \no{\sqrt{\rxl}^{-1} \varrho_x \sqrt{\rxl}^{-1}} \, .
\eea
Let us denote $\N_\ell$ the MCM when the pre-measurement information $\ell$ is given. Just as MCM without side information, $\sum_{x \in X_\ell} \N_\ell(x)$  does not always form an identity. In such a case, we need to include an additional effect, $\N_\ell(0)$, such that $\sum_{x \in X_\ell} \N_\ell (x)+\N_\ell (0)=\id$.  When $\varrho_x$ is a pure state, $\N_\ell(x)$ is a rank-1 operator, 
\bea
\N_\ell(x)=a_x \widetilde{\Pi}_x \label{eq:premcm}
\eea
where
\bea
\widetilde{\Pi}_x=\frac{\rxl^{-1}\varrho_x \rxl^{-1}}{\tr{\rxl^{-1}\varrho_x \rxl^{-1}}} \,  \label{eq:premcmproj}
\eea
and $0\leq a_x \leq 1$ is a parameter to be further optimized. If $\rxl$ is not invertible, then we take pseudoinverse instead. 
When $\varrho_x$ is prepared, the measurement outcome $y$ is conclusive if $y \in X_\ell$. 
Therefore, the optimal conclusive rate given side information $X_\ell$ is
\bea
R_\ell=\max_{\N_\ell}\tr{\rxl \sum_{y \in X_\ell} \N_\ell (y)} \, .
\eea
The overall optimal conclusive rate with pre-measurement information is
\bea
\rpre=\sum_\ell \qxl R_\ell \, .
\eea
MCM with pre-measurement information in \eqref{eq:premcm}is found by considering $\qxl\rxl$, the unnormalized state for subset $X_\ell$, instead of $\varrho$ as in \eqref{eq:mcm}, the ensemble state for the set $X$. 

%%%%%%%%%%%%%%%%%%%%%%%%
\subsection{Post-measurement information without metainformation}
%%%%%%%%%%%%%%%%%%%%%

In post-measurement information scenario, Bob makes a MCM $\N$ defined in \eqref{eq:mcm} and post-processes the measurement outcome with side information to guess $x \in X_\ell$ with maximum confidence.  
Given a measurement outcome $y$ and side information $\ell$, the conditional probability of $x\in X_\ell$ is
\begin{equation}
\Pr(x |y,\ell)=\frac{q_x \tr{\varrho_x \N(y)}}{\qxl \tr{\rxl \N(y)}} \, .
\end{equation}
Since we are considering the maximum confidence, we only consider the highest conditional probability for each $x$. Therefore, the maximum confidence of $x$ with post-measurement information is defined as
\bea
\cpost(x)=\max_{y} \Pr(x|y,\ell)=\max_y \frac{q_x \tr{\varrho_{x} \N(y)}}{\qxl \tr{\rxl \N(y)}} \, .\label{eq:cpost}
\eea
We denote 
\bea
Y_x = \{ y \in Y : \Pr(x|y,\ell) = \cpost(x) \} \, . \label{eq:ysubset}
\eea
From the receiver's viewpoint the set of conclusive outcomes, when given the side information $\ell$, is $Y_\ell \equiv \cup_{x \in X_{\ell}} Y_x$.
When the receiver records one of the conclusive outcomes and gets side information $\ell$, he relabels $y$ to $x$ for which $y\in Y_x$.

The conclusive rate is
\bea
\rpost=\sum_\ell  \qxl \tr{\rxl \sum_{y \in Y_{\ell}} \N(y)} \, .\label{eq:rpost}
\eea
We remark that the conclusive rate in \eqref{eq:rpost} can be lower than the optimal conclusive rate without side information in \eqref{eq:conclusive}.

%%%%%%%%%%%%%%%%%%%%
\subsection{Post-measurement information with metainformation}
%%%%%%%%%%%%%%%%%%%%

Suppose that Bob knows about the partition. Recall that the MCM projections for pre-measurement case in \eqref{eq:premcmproj} can be implemented with a single measurement if he knows about the partition. Namely, a measurement defined as
\bea
\N(x)=a_x \widetilde{\Pi}_x \label{eq:mcmmeta},
\eea
where $0 \leq a_x \leq 1$ is a free parameter to be further optimized and $\widetilde{\Pi}_x$ is defined in \eqref{eq:premcmproj},is the MCM with metainformation. Therefore, when the figure of merit is confidence, the post-measurement information with metainformation is always as good as pre-measurement information, i.e., 
$$
\cmeta(x)=\cpre(x)
$$
for all $x\in X$.

However, the optimal conclusive rate can be different. 
This is because Bob needs to implement the projections $\widetilde{\Pi}_x$ simulatneously for all $x \in X$, wheares in pre-measurement information case he only needs to implement $\widetilde{\Pi}_x$ for $x \in X_\ell$.  A measurement outcome $y$ with side information $\ell$ is conclusive if $y \in X_\ell$. The conclusive rate with metainformation is then 
\bea
\rmeta=\max_\N \sum_\ell \qxl \tr{\rxl \sum_{y \in X_\ell } \N(y)} \label{eq:rmeta}
\eea
where $\N$ is given in \eqref{eq:mcmmeta}. 
Therefore,
\bea
\rpre \geq \rmeta \, . \nonumber
\eea
We call the measurement defined in \eqref{eq:mcmmeta} that gives $\rmeta$ the \emph{anticipative} MCM. Although the confidence obtainable with pre-measurement information can always be achieved with post-measurement information with metainformation, pre-measurement information can yield a higher conclusive rate.
\subsection{Examples}

In this Subsection, we provide several examples that illustrate the advantage of having metainformation in maximum-confidence discrimination. These cases show that metainformation can be entirely useless, somewhat useful but not as effective as pre-measurement information, or useful as pre-measurement information.

\subsubsection{BB84 states: $\cpre=\cmeta=\cpost>C$ and $\rpre>\rmeta=\rpost$ case.}
Let us consider the symbol set $X=\{1+,1-,2+,2-\}$ with the partition $X=X_1\cup X_2$ where $X_\ell=\{\ell+,\ell-\}$. The symbols are encoded in $\sigma_x$ and $\sigma_z$ basis as 
\bea
\varrho_{1\pm}=\frac{1}{2}(\id \pm  \sigma_z ), \varrho_{2\pm}=\frac{1}{2}(\id \pm  \sigma_x ) \, , \nonumber
\eea
and the side information is the value of the basis, $\ell$. See Fig. \ref{fig:bb84} (a).  The task is to discriminate the states with the maximum confidence with side information. 

Suppose $x \in X_\ell$ is chosen. Without side information, the optimal MCM is a randomly chosen projective measurement related to $\sigma_x$ and $\sigma_z$,
\bea
\N(j\pm)=t_j \varrho_{j\pm} \, , \label{eq:bb84mcm}
\eea
where $t_j$ are convex coefficients, and any of these measurements yield 
$$C(x)=\frac{1}{2},  R=1 \,.$$ 
Let us first consider the pre-measurement information case. The side information $\ell$ is sent to Bob before making a measurement, and he will make the optimal MCM for $X_\ell$. Since $X_\ell$ contains only orthogonal states for both $\ell=1,2$, he can perfectly discriminate all symbols in $X$. In other words, these measurements yield 
$$\cpre(x)=1, \rpre=1\,.$$ 

The second scenario is when the side information $\ell$ is sent after a measurement. Bob makes the standard MCM in \eqref{eq:bb84mcm} and he will post-process measurement outcomes with the side information. The subsets in \eqref{eq:ysubset} are $Y_x=\{x\}$, and the measurement outcome $y$ is conclusive with given side information $\ell$ if $y \in Y_\ell=\{\ell+,\ell-\}$. In other words, a measurement outcome is conclusive if the bases used by Alice and Bob match. When the outcomes are conclusive, the measurement with post-processing yields 
$$\cpost(x)=1 \,. $$
The conclusive rate with post-measurement information is
$$
\rpost=\sum_\ell  \qxl \tr{\rxl \sum_{y \in Y_{\ell}} \N(y)}=\frac{1}{2}
$$
for any values of $t_j$ in \eqref{eq:bb84mcm}. This is the probability that Alice and Bob use the same basis. 

The last scenario is when Bob gets post-measurement information with metainformation, i.e., he knows about the partition. In this case, however, metainformation does not provide any advantage for confidence or conclusive rate. The MCM projections in \eqref{eq:premcmproj} are $\widetilde{\Pi}_x=\varrho_x$ and with these projections the optimal conclusive rate is 
\bea
\rmeta=\max_{\{a_y\}} \sum_\ell \qxl \tr{\rxl \sum_{y \in X_\ell} a_y \widetilde{\Pi}_y} =\frac{1}{2} \nonumber
\eea
In summary, 
\bea\nonumber
\cpre(x)=\cmeta(x)=\cpost(x),~
\rpre>\rmeta=\rpost \, .
\eea
We remark that the scenario of the BB84 protocol can be interpreted as a task of maximum-confidence discrimination with post-measurement information, and metainformation is not useful in this case. 

\begin{figure}[]
  \centering
  \begin{subfigure}[b]{0.45\textwidth}
    \centering
    \includegraphics[width=\linewidth]{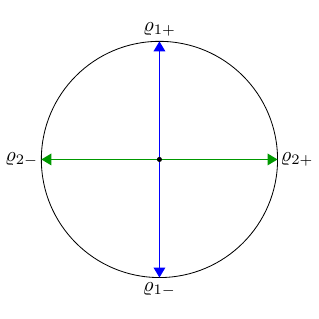}
    \caption*{(a) BB84 states }
  \end{subfigure}
  \hspace{0.04\textwidth}
  \begin{subfigure}[b]{0.45\textwidth}
    \centering
    \includegraphics[width=\linewidth]{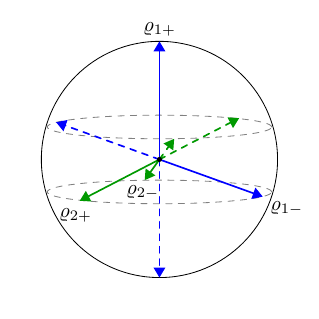}
    \caption*{(b) Tetrahedron states }
  \end{subfigure}
  \caption{Two possible ways of encoding symbols $X=\{1\pm, 2\pm\}$. In (a), the symbols are encoded in $\sigma_x$ and $\sigma_z$ basis as in the BB84 protocol. In this case, the anticipative measurement and the standard measurement are identical. The metainformation is not useful, i.e., $\cmeta=\cpost$ and $\rmeta=\rpost$, and the states can be unambiguously identified with post-measurement information. In (b), the symbols are encoded in states that form a tetrahedron. For this encoding, the standard measurement is the same as the states, wheares the anticipative measurement is the orthogonal complement of the states, depicted by dashed arrows. The metainformation offers advantage $\cmeta>\cpost$ and the states can be unambiguously identified with post-measurement information only if metainformation is provided. }
  \label{fig:bb84}
\end{figure}

\subsubsection{BB84-like task with tetrahedron states: $\cpre=\cmeta>\cpost>C$ and $\rpre>\rmeta$ case.}
One may instead use a different ensemble to do BB84-like task with the same symbol set $X=\{1+,1-,2+,2-\}$.  For instance, consider tetrahedron states $\varrho_x=\ketbra{\psi_x}$  defined as
\begin{align*}
\ket{\psi_{1+}}&=\ket{0}, ~
\ket{\psi_{1-}}=\frac{1}{\sqrt{3}}\ket{0}+\sqrt{\frac{2}{3}}\ket{1}\\
\ket{\psi_{2\pm}}&=\frac{1}{\sqrt{3}}\ket{0}+e^{\pm\frac{2\pi i}{3}}\sqrt{\frac{2}{3}}\ket{1} \, ,
\end{align*}
see Fig. \ref{fig:bb84} (b). 
When there is no classical side information, the optimal MCM is 
$$\N(x)=\frac{1}{2}\varrho_x \, .$$ 
This measurement yields $C(x)=\frac{1}{2}$ and $R=1$, just as in the previous example. However, this ensemble will show that metainformation can be useful in maximum-confidence discrimination with post-measurement information. 

Firstly, let us consider the pre-measurement information scenario. The partition is $X=X_1 \cup X_2$ where $X_\ell=\{\ell_+,\ell-\}$. Based on $\ell$, Bob makes the optimal MCM for $X_\ell$ with the projections defined as
\bea
\widetilde{\Pi}_{\ell \pm}=\varrho^\perp_{\ell \mp} \, . \label{eq:tetraproj}
\eea
The optimal conclusive rate for $X_\ell$ is obtained from
\bea
R_\ell=\max_{\{a_{\ell\pm}\}} \tr{\rxl (a_{\ell +} \widetilde{\Pi}_{\ell +}+a_{\ell -} \widetilde{\Pi}_{\ell -})}=1-\frac{1}{\sqrt{3}} \nonumber
\eea
These measurements yield
$$\cpre(x)=1, ~\rpre=1-\frac{1}{\sqrt{3}} \, .$$ 
Now suppose that the side information is sent after the measurement. The subsets in \eqref{eq:ysubset} are $Y_x=\{x\}$, and with post-processing the measurement yields
\bea
\cpost(x)=\frac{3}{4},~ \rpost=\frac{2}{3} \nonumber
\eea
Lastly, suppose Bob gets metainformation along with post-measurement information. Knowing about the partition, he can implement all projections in \eqref{eq:tetraproj} with a single measurement $\N$ whose effects are defined as $\N(\ell\pm)=a_{\ell\pm} \varrho_{\ell\mp}^\perp$ where $a_{\ell \pm}$ are parameters to be determined. 
Since these projections are the same projections for MCM with pre-measurement information,
$$
\cmeta(x)=1 \,.
$$
The optimal conclusive rate can be obtained by finding the optimal parameters $a_{\ell \pm}$

\bea
\rmeta=\max_{\{a_{\ell \pm } \}} \sum_\ell \qxl \tr{\rxl ( a_{\ell +} \varrho_{\ell -}^\perp+a_{\ell -} \varrho_{\ell +}^\perp)} =\frac{1}{3} \nonumber
\eea
Therefore,
\begin{align*}
    \cpre(x)=\cmeta(x)>\cpost(x),~ \rpre>\rmeta \, .
\end{align*}
This example demonstrates that one can achieve higher confidence with post-measurement information if metainformation is provided but cannot attain the conclusive rate with pre-measurement information.

\begin{figure}
  \centering
  \begin{subfigure}[b]{0.45\textwidth}
    \centering
    \includegraphics[width=\linewidth]{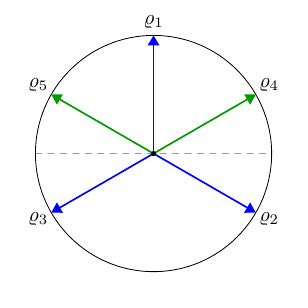}
    \caption*{ (a) Ensemble with the partition}
  \end{subfigure}
  \hspace{0.04\textwidth}
  \begin{subfigure}[b]{0.45\textwidth}
    \centering
    \includegraphics[width=\linewidth]{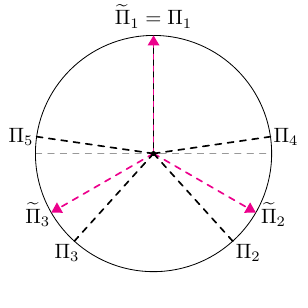}
    \caption*{ (b) MCMs with side information }
  \end{subfigure}
  \caption{ In (a), symbols $X=\{1,2,3,4,5\}$ are encoded in quantum states with partition $X_1=\{1,2,3\}$ and $X_2=\{4,5\}$. For both subsets, the optimal MCM consists of effects that are identical to the state ensemble for $X_1$, marked as blue arrows. In (b), the projections for the anticipative MCM and the standard MCM are depicted as purple arrows and black dashed lines, respectively. }
  \label{fig:mcmex3}
\end{figure}

\subsubsection{$\cpre=\cmeta\geq\cpost>C$ and $\rpre=\rmeta$ case}
The previous examples show that metainformation can either be irrelevant or offer higher confidence than post-measurement information alone. In contrast, the final example demonstrates that when metainformation is available, post-measurement information can be just as valuable as pre-measurement information.

Consider a symbol set $X=\{1,2,3,4,5\}$ with partition $X_1=\{1,2,3\}$ and $X_2=\{4,5\}$. The symbols are encoded in the following states,
\begin{align*}
&\varrho_{x}=\frac{1}{2}(\id+\sin\frac{2 \pi x}{3} \sigma_x + \cos\frac{2\pi x}{3}\sigma_z), ~x \in X_1 \, ,\\
&\varrho_4=\varrho_3^\perp \, , \quad  \varrho_5=\varrho_2^\perp \, ,
\end{align*} 
as illustrated in Fig. \ref{fig:mcmex3} (a). The projections for the standard MCM are obtained from \eqref{eq:mcmprojector},  
\bea
\Pi_1=\frac{1}{2}(I+\sigma_z), \Pi_{2,3}=\frac{1}{2}(I\pm\frac{12 \sqrt{3}}{31} \sigma_x - \frac{23}{31}\sigma_z), \Pi_{4,5}=\frac{1}{2}(I\pm \frac{4 \sqrt{3}}{7} \sigma_x + \frac{1}{7}\sigma_z) \, . \label{eq:eq3mcmproj}
\eea
These projectors yield the confidence $C(1)=\frac{1}{3}, C(2)=C(3)=\frac{11}{24}, C(4)=C(5)=\frac{3}{8}$. 

Now let us consider pre-measurement information case. The MCM projections with pre-measurement information are
\bea
\widetilde{\Pi}_1=\varrho_1, \widetilde{\Pi}_2=\widetilde{\Pi}_4=\varrho_2, \widetilde{\Pi}_3=\widetilde{\Pi}_5=\varrho_3 \, . \label{eq:ex3preproj}
\eea
One can show that the optimal MCMs with pre-measurement information are $\N_\ell (x)=\frac{2}{3}\widetilde{\Pi}_x$ for both $X_1$ and $X_2$. These measurements yield
$$
\cpre(x)=\frac{2}{3}, \forall x \in X_1,  \cpre(x)=1, \forall x \in X_2
$$
and 
$$
\rpre=\frac{4}{5} \, .
$$

Now let us find the maximum confidence by with post-measurement information. The relabelling subsets in \eqref{eq:ysubset} are
\begin{align*}
&Y_x=\{x\} \text{ if } x \in X_\ell, \\
&Y_4=\{2\}, Y_5=\{3\} \, .
\end{align*}
The MCM projections in \eqref{eq:eq3mcmproj} with this post-processing yields the confidence
$$
\cpost(1)=\frac{2}{3}, \cpost(2)=\cpost(3)=\frac{121}{186}, \cpost(4)=\cpost(5)=\frac{75}{78} \, .
$$ 
If Bob knows about the partition, he implements the projections in \eqref{eq:ex3preproj} with a single measurement, and the optimal conclusive rate is
$$
\rmeta=\frac{4}{5} \, .
$$
Therefore, metainformation can make post-measurement information as useful as pre-measurement information, 
$$
\cpost=\cmeta>\cpost , \quad  \rpre=\rmeta \, .
$$

%%%%%%%%%%%%%%%%
\section{Summary}\label{sec:summary}
%%%%%%%%%%%%%%%%%%

When transmitting information, it can be in the form of classical information, quantum information, or some of their combination.
We have considered an amalgamation where the primary message is encoded into quantum states while there is classical side information that gives partial information on the message.

We have shown that the effectiveness of classical side information depends on several factors in quantum guessing tasks that optimize different figures of merits: success probability and confidence. Specifically, even when the side information is revealed only after the measurement had been performed, two fundamentally distinct scenarios emerges. 
In one case, the receiver is aware in advance that some form of additional information would follow. 
In the other, the receiver has no prior expectation of receiving further information. 
This distinction influences both the strategies employed and the achievable performance in the guessing tasks.

The information about the form of additional information is not information in the usual sense; hence we have called it metainformation.
Our investigation demonstrates that metainformation can make a crucial difference. Since metainformation is not information in the usual sense, it may not be obvious to identify when it is present, and when not. The role of metainformation should be investigated in other information processing scenarios.

\section*{Acknowledgements}
The authors acknowledge the financial support from the Business Finland project BEQAH.

\bibliographystyle{unsrt}

\bibliography{bibliography}

\end{document}